\documentclass{easychair}
\usepackage{isabelle}
\usepackage{isabellesym}
\usepackage{xcolor}

\newcommand\suc[1]{\isamath{#1+1}}
\newenvironment{isadisplay}{%
  \begin{list}{}{%
    \setlength{\leftmargin}{0.25em}%
    \setlength{\rightmargin}{\leftmargin}}\item}%
{\end{list}}
\newcommand\minat[2]{\named{min}\subed{#2}\pared{#1}}
\newcommand\repl[3]{\ensuremath{#2\langle#1\rangle#3}}
\newcommand\strong{\ensuremath{\wqole_{\mathsf{o}}}}
\newcommand\weaklt{\ensuremath{\lhd}}
\newcommand\weakle{\ensuremath{\unlhd}}
\newcommand\suffix{\ensuremath{<}}
\newcommand\vals[1]{\ensuremath{#1^{\mathsf{o}}}}

\newcommand\wqolt{\ensuremath{\prec}}
\newcommand\wqole{\ensuremath{\preceq}}
\newcommand\lemb{\ensuremath{\wqole_{\mathsf{emb}}}}
\newcommand\hemb{\ensuremath{\wqole_{\mathsf{emb}}}}
\newcommand\embhemb{\ensuremath{\wqole^=_{\mathsf{emb}}}}
\newcommand\subtree[1][]{\ensuremath{\lhd}}
\newcommand\subtreebin[3][]{#2 \subtree[#1] #3}
\newcommand\nwqole{\ensuremath{\not\preceq}}
\newcommand\secref[1]{Section~\ref{sec:#1}}
\newcommand\lemref[1]{Lemma~\ref{lem:#1}}
\newcommand\thmref[1]{Theorem~\ref{thm:#1}}
\newcommand\defref[1]{Definition~\ref{def:#1}}
\isabellestyle{it}
\renewcommand{\isastyletxt}{\isastyletext}
\renewcommand{\isastyle}{\isastyleminor}
\newcommand\ceta{\textsf{C\kern-0.2exe\kern-0.5exT\kern-0.5exA}}
\newcommand\isafor{\textsf{Isa\kern-0.2exF\kern-0.2exo\kern-0.2exR}}
\newcommand\Emb{\mathcal{E}\textsf{mb}(\wqole)}

\newcommand\pared[1]{\isacharparenleft#1\isacharparenright}

\newcommand\subed[1]{\isactrlbsub#1\isactrlesub}

\newcommand\named[1]{\isamath{\mathit{#1}}}
\newcommand\caled[1]{\isamath{\mathcal{#1}}}
\newcommand\reflon[2]{\named{refl}\subed{#2}\pared{#1}}
\newcommand\transon[2]{\named{trans}\subed{#2}\pared{#1}}
\newcommand\afon[2]{\named{af}\subed{#2}\pared{#1}}
\newcommand\wqoon[2]{\named{wqo}\subed{#2}\pared{#1}}
\newcommand\wfon[2]{\named{wf}\subed{#2}\pared{#1}}
\newcommand\good[2]{\named{good}\subed{#1}\pared{#2}}
\newcommand\bad[2]{\named{bad}\subed{#1}\pared{#2}}
\newcommand\lists[1]{\ensuremath{#1^*}}
\newcommand\trees[2][]{\caled{T}\pared{#2}}
\newcommand\treeslist[2][]{\caled{T}\subed{list}\pared{#2}}

\theoremstyle{definition}
\newtheorem{definition}{Definition}
\theoremstyle{plain}
\newtheorem{theorem}{Theorem}
\newtheorem{lemma}{Lemma}

\title{A Locale for Minimal Bad Sequences}
\titlerunning{A Locale for Minimal Bad Sequences}
\author{Christian Sternagel\thanks{%
  Supported by the Austrian Science Fund (FWF): J3202.}}
\authorrunning{C.~Sternagel}
\institute{JAIST, Japan\\
  \email{c-sterna@jaist.ac.jp}}

\definecolor{dblue}{rgb}{0.2,0.2,0.7}%
\hypersetup{%
  colorlinks=true,%
  pdftex,%
  pdfborder={0 0 0},%
  linkcolor=dblue,%
  citecolor=dblue,%
  urlcolor=dblue,%
}

\begin{document}
\maketitle

\begin{abstract}
We present a locale that abstracts over the necessary ingredients for
constructing a minimal bad sequence, as required in classical proofs of Higman's
lemma and Kruskal's tree theorem.
\end{abstract}

\bibliographystyle{plain}

\section{Introduction}

The so called \emph{minimal bad sequence argument} was first used by
Nash-Williams in \cite{Nash-Williams1963}; where he first proves a variant of
\emph{Higman's lemma} \cite{Higman1952} for finite sets, and then -- again using
a minimal bad sequence argument -- \emph{Kruskal's tree theorem}.  This proof is
usually considered to be simple and elegant. To a certain extend we agree, but
then again, formalizing a proof (using a proof assistant) typically requires us
to be more rigorous than on paper. During our Isabelle/HOL \cite{Isabelle}
formalization of Higman's lemma and Kruskal's tree theorem \cite{afp-wqo} we
found that Nash-Williams' reasoning for constructing a minimal bad sequence is
far from comprehensive. That is, assuming that there exists a minimal bad
sequence, both proofs can be formalized almost exactly as presented in
\cite{Nash-Williams1963}.%
\footnote{Apart from two claims of the form ``the absence of a bad sequence of
a certain shape implies the absence of any bad sequence,'' whose proofs are
omitted.}
But to prove the existence of such a minimal bad sequence turns out to be rather
involved. (A step that is omitted in any classical paper-proof using the minimal
bad sequence argument we could catch sight of.)

To this end, we formalized a locale $\named{mbs}$ that encapsulates the required
ingredients for constructing a minimal bad sequence starting from an arbitrary
bad sequence. Interpreting this locale for lists (with list-embedding)
and finite trees (with homeomorphic embedding on finite trees) is easy and takes
care of the biggest ``gaps'' in the paper-proofs of Nash-Williams.

The remainder is structured as follows: In \secref{prelim} we give some
preliminaries on well-quasi-order (wqo) theory and put our minimal bad sequence
construction into context. Afterwards, in \secref{mbs}, we present our locale
and sketch our construction of a minimal bad sequence. Then, in \secref{appl},
we give two applications. Finally, we conclude in
\secref{concl}.

\section{\label{sec:prelim}Preliminaries}

In this section we recall well-quasi-orders as well as Higman's lemma and
Kruskal's tree theorem. This serves to give a context for our minimal bad
sequence construction. Moreover, we give basic definitions that are used
throughout our formalization.

In the following we use $\wqole$ for an arbitrary (not necessarily reflexive)
binary relation on the
elements of an arbitrary set $A$ (when not stated otherwise). In the literature,
well-quasi-orders are typically defined as follows:
\begin{definition}\label{def:wqo}
A set $A$ is \emph{well-quasi-ordered} (wqo) by $\wqole$ iff $\wqole$ is
reflexive, transitive, and satisfies: for every infinite sequence $f$ over
elements of $A$ there are indices $i < j$, s.t., $f~i \wqole f~j$.
\end{definition}
Using the terminology of \cite{Nash-Williams1963} an infinite sequence $f$ for
which we have indices $i < j$, s.t., $f~i \wqole f~j$ is called \emph{good}. A
sequence that is not good is called \emph{bad}.  Moreover, a relation $\wqole$
satisfying the last condition of \defref{wqo} (i.e., all infinite sequences over
elements of $A$ are good) is called \emph{almost full}.%
\footnote{The notion \emph{almost full} was first introduced in
\cite{Veldman&Benzem1993} and very recently revived in \cite{almost-full}.}

\begin{isabellebody}%
\def\isabellecontext{Content}%
\isadelimtheory
\endisadelimtheory
\isatagtheory
\endisatagtheory
{\isafoldtheory}%
\isadelimtheory
\endisadelimtheory
\isadelimproof
\endisadelimproof
\isatagproof
\endisatagproof
{\isafoldproof}%
\isadelimproof
\endisadelimproof
\begin{isamarkuptext}%
Now consider the above concepts as defined in our Isabelle/HOL formalization,
where we encode binary relations by functions of type \isa{\isasymalpha\ {\isaliteral{5C3C52696768746172726F773E}{\isasymRightarrow}}\ \isasymalpha\ {\isaliteral{5C3C52696768746172726F773E}{\isasymRightarrow}}\ bool}
(i.e., binary predicates) and infinite sequences by functions of type \isa{nat\ {\isaliteral{5C3C52696768746172726F773E}{\isasymRightarrow}}\ \isasymalpha}.

An infinite sequence \isa{f} is \emph{good} (w.r.t.\ a binary predicate \isa{\wqole}), written \isa{\good{\wqole}{f}}, iff \isa{{\isaliteral{5C3C6578697374733E}{\isasymexists}}i\ j{\isaliteral{2E}{\isachardot}}\ i\ {\isaliteral{3C}{\isacharless}}\ j\ {\isaliteral{5C3C616E643E}{\isasymand}}\ f\ i\ \wqole\ f\ j}.
It is \emph{bad} (w.r.t.\ \isa{\wqole}), written \isa{\bad{\wqole}{f}}, iff it is not
good.
A binary relation \isa{\wqole} is \emph{reflexive on} a set \isa{A}, written
\isa{\reflon{\wqole}{A}}, iff \isa{{\isaliteral{5C3C666F72616C6C3E}{\isasymforall}}a{\isaliteral{5C3C696E3E}{\isasymin}}A{\isaliteral{2E}{\isachardot}}\ a\ \wqole\ a}.
It is \emph{transitive on} \isa{A}, written \isa{\transon{\wqole}{A}}, iff \isa{{\isaliteral{5C3C666F72616C6C3E}{\isasymforall}}a{\isaliteral{5C3C696E3E}{\isasymin}}A{\isaliteral{2E}{\isachardot}}\ {\isaliteral{5C3C666F72616C6C3E}{\isasymforall}}b{\isaliteral{5C3C696E3E}{\isasymin}}A{\isaliteral{2E}{\isachardot}}\ {\isaliteral{5C3C666F72616C6C3E}{\isasymforall}}c{\isaliteral{5C3C696E3E}{\isasymin}}A{\isaliteral{2E}{\isachardot}}\ a\ \wqole\ b\ {\isaliteral{5C3C616E643E}{\isasymand}}\ b\ \wqole\ c\ {\isaliteral{5C3C6C6F6E6772696768746172726F773E}{\isasymlongrightarrow}}\ a\ \wqole\ c}.
It is \emph{almost full on} \isa{A}, written \isa{\afon{\wqole}{A}}, iff
\isa{{\isaliteral{5C3C666F72616C6C3E}{\isasymforall}}f{\isaliteral{2E}{\isachardot}}\ {\isaliteral{28}{\isacharparenleft}}{\isaliteral{5C3C666F72616C6C3E}{\isasymforall}}i{\isaliteral{2E}{\isachardot}}\ f\ i\ {\isaliteral{5C3C696E3E}{\isasymin}}\ A{\isaliteral{29}{\isacharparenright}}\ {\isaliteral{5C3C6C6F6E6772696768746172726F773E}{\isasymlongrightarrow}}\ \good{\wqole}{f}}.
The relation \isa{\wqole} is a \emph{wqo on} \isa{A}, written \isa{\wqoon{\wqole}{A}}, iff \isa{\transon{\wqole}{A}\ {\isaliteral{5C3C616E643E}{\isasymand}}\ \afon{\wqole}{A}}.
It \emph{is well-founded on} \isa{A}, written \isa{\wfon{\wqole}{A}}, iff \isa{{\isaliteral{5C3C6E6F743E}{\isasymnot}}\ {\isaliteral{28}{\isacharparenleft}}{\isaliteral{5C3C6578697374733E}{\isasymexists}}f{\isaliteral{2E}{\isachardot}}\ {\isaliteral{5C3C666F72616C6C3E}{\isasymforall}}i{\isaliteral{2E}{\isachardot}}\ f\ i\ {\isaliteral{5C3C696E3E}{\isasymin}}\ A\ {\isaliteral{5C3C616E643E}{\isasymand}}\ f\ {\isaliteral{28}{\isacharparenleft}}\suc{i}{\isaliteral{29}{\isacharparenright}}\ \wqole\ f\ i{\isaliteral{29}{\isacharparenright}}} (i.e., there are no infinite descending sequences
over elements of \isa{A}).
It is easy to see that every almost full relation is also reflexive.
\begin{lemma}\label{lem:af imp refl}
\isa{\afon{\wqole}{A}\ {\isaliteral{5C3C4C6F6E6772696768746172726F773E}{\isasymLongrightarrow}}\ \reflon{\wqole}{A}}
\end{lemma}
\begin{proof}
Take an arbitrary but fixed \isa{x\ {\isaliteral{5C3C696E3E}{\isasymin}}\ A}. We have to show \isa{x\ \wqole\ x}.
To this end, consider the infinite sequence \isa{f\ i\ {\isaliteral{3D}{\isacharequal}}\ x}, which is an
infinite sequence over elements of \isa{A} and thus, since \isa{\wqole} is almost
full, we obtain indices \isa{i\ {\isaliteral{3C}{\isacharless}}\ j} with \isa{f\ i\ \wqole\ f\ j}. Since
\isa{f} equals \isa{x} at every position, we obtain the required \isa{x\ \wqole\ x}.
\end{proof}
Thus, any almost full relation that is transitive is a wqo and the other way
round.  Since, in our formalization, we strive for minimality, and furthermore
transitivity is not required in the proofs of Higman's lemma and Kruskal's tree
theorem (and typically easily added afterwards), we concentrate on almost full
relations.

Before we continue, note that wqos are interesting (at least) due to the
following fact: The strict part \isa{\wqolt} of a wqo \isa{\wqole} is
well-founded on \isa{A}.  Where \isa{{\isaliteral{22}{\isachardoublequote}}x\ \wqolt\ y\ {\isaliteral{3D}{\isacharequal}}\ {\isaliteral{28}{\isacharparenleft}}x\ \wqole\ y\ {\isaliteral{5C3C616E643E}{\isasymand}}\ y\ \nwqole\ x{\isaliteral{29}{\isacharparenright}}{\isaliteral{22}{\isachardoublequote}}}. 
\begin{lemma}
\isa{\wqoon{\wqole}{A}\ {\isaliteral{5C3C4C6F6E6772696768746172726F773E}{\isasymLongrightarrow}}\ \wfon{\wqolt}{A}}
\end{lemma}
\begin{proof}
Assume to the contrary that \isa{\wqolt} is not well-founded on \isa{A}.
Then there is an infinite descending sequence \isa{f} over elements of \isa{A}.  Hence, for all \isa{i\ {\isaliteral{3C}{\isacharless}}\ j} we have \isa{f\ j\ \wqolt\ f\ i},
since \isa{\wqolt} is transitive. Moreover, since \isa{\wqolt} is
irreflexive, \isa{f\ i\ \nwqole\ f\ j} for all \isa{i\ {\isaliteral{3C}{\isacharless}}\ j}. Thus \isa{f}
is bad, contradicting the fact that \isa{\wqole} is almost full.
\end{proof}

Now, consider Higman's lemma and Kruskal's tree theorem as stated in our
formalization.

\begin{lemma}[Higman's Lemma]\label{lem:higman}
\isa{\wqoon{\wqole}{A}\ {\isaliteral{5C3C4C6F6E6772696768746172726F773E}{\isasymLongrightarrow}}\ \wqoon{\lemb}{\lists{A}}}
\end{lemma}

\begin{theorem}[Kruskal's Tree Theorem]\label{thm:kruskal}
\isa{\wqoon{\wqole}{A}\ {\isaliteral{5C3C4C6F6E6772696768746172726F773E}{\isasymLongrightarrow}}\ \wqoon{\hemb}{\trees[Node]{A}}}
\end{theorem}
In the above two statements, \isa{\lists{A}} denotes the set of finite lists
built over elements of \isa{A} and \isa{\trees[Node]{A}} denotes the set of finite
trees built over elements of \isa{A}. The binary relation \isa{\lemb},
denotes homeomorphic embedding on finite lists and finite trees, respectively
(see \secref{appl} for concrete definitions). (Note that the interesting parts
of the above proofs correspond to \isa{\afon{\wqole}{A}\ {\isaliteral{5C3C4C6F6E6772696768746172726F773E}{\isasymLongrightarrow}}\ \afon{\lemb}{\lists{A}}} and \isa{\afon{\wqole}{A}\ {\isaliteral{5C3C4C6F6E6772696768746172726F773E}{\isasymLongrightarrow}}\ \afon{\hemb}{\trees[Node]{A}}}, respectively.)

In both proofs (as presented by Nash-Williams) the existence of a minimal bad
sequence is essential. However, the only thing Nash-Williams has to say about
the construction of a minimal bad sequence is roughly (where we use $\vals{A}$
to denote the set of ``objects'' built over elements of $A$; which might refer
to the set of finite subsets, the set of finite lists, the set of finite trees,
\ldots in a concrete case):
\begin{quote}\sl
Select an $x_1 \in \vals{A}$ such that $x_1$ is the first term of a bad sequence
of members of $\vals{A}$ and $|x_1|$ is as small as possible. Then select an
$x_2$ such that $x_1$, $x_2$ (in that order) are the first two terms of a bad
sequence of members of $\vals{A}$ and $|x_2|$ is as small as possible [\ldots].
Assuming the Axiom of Choice, this process yields a [minimal] bad sequence
[\ldots]
\end{quote}
Interestingly, most non-formalized proofs of Higman's lemma and Kruskal's tree
theorem in the literature (that the author is aware of) are similarly vague
about the actual construction of a minimal bad sequence.  The point is that a
crucial proof is missing in the above recipe, namely that it is actually
possible to select elements as described.

In the next section we make the notion \emph{minimal bad sequence} more concrete
(i.e., answer the question: ``Minimal in what sense?'') but at the same time
abstract over the basic ingredients (``What are the properties that have to be
satisfied for a minimal bad sequence to exist?'').%
\end{isamarkuptext}%
\isamarkuptrue%
\isamarkupsection{\label{sec:mbs}Constructing Minimal Bad Sequences%
}
\isamarkuptrue%
\begin{isamarkuptext}%
We encapsulate the construction of a minimal bad sequence over elements from a
given set (which we call \emph{objects}) inside a locale taking the following
arguments:
\begin{itemize}
\item
a function \isa{\vals{A}} that returns the set of objects that are built over
elements of \isa{A},
\item
a relation \isa{\strong} that is used to check whether an infinite sequence
of objects is good (where $\wqole$ is a relation on elements of \isa{A}),
\item
and a relation \isa{\weaklt} used for checking minimality (whose reflexive
closure is denoted by \isa{\weakle}).
\end{itemize}
The required properties are:
\begin{itemize}
\item
\emph{right-compatibility} of \isa{\weaklt} with \isa{\strong}:
\isa{{\isaliteral{5C3C6C6272616B6B3E}{\isasymlbrakk}}x\ \strong\ y{\isaliteral{3B}{\isacharsemicolon}}\ y\ \weaklt\ z{\isaliteral{5C3C726272616B6B3E}{\isasymrbrakk}}\ {\isaliteral{5C3C4C6F6E6772696768746172726F773E}{\isasymLongrightarrow}}\ x\ \strong\ z}
\item
\emph{well-foundedness} of \isa{\weaklt} on elements of \isa{\vals{A}}:
\isa{\wfon{\weaklt}{\vals{A}}}
\item
\emph{transitivity} of \isa{\weaklt}:
\isa{{\isaliteral{5C3C6C6272616B6B3E}{\isasymlbrakk}}x\ \weaklt\ y{\isaliteral{3B}{\isacharsemicolon}}\ y\ \weaklt\ z{\isaliteral{5C3C726272616B6B3E}{\isasymrbrakk}}\ {\isaliteral{5C3C4C6F6E6772696768746172726F773E}{\isasymLongrightarrow}}\ x\ \weaklt\ z}
\item
\isa{\weaklt} \emph{reflects} the property of being in \isa{\vals{A}}:
\isa{{\isaliteral{5C3C6C6272616B6B3E}{\isasymlbrakk}}x\ \weaklt\ y{\isaliteral{3B}{\isacharsemicolon}}\ y\ {\isaliteral{5C3C696E3E}{\isasymin}}\ \vals{A}{\isaliteral{5C3C726272616B6B3E}{\isasymrbrakk}}\ {\isaliteral{5C3C4C6F6E6772696768746172726F773E}{\isasymLongrightarrow}}\ x\ {\isaliteral{5C3C696E3E}{\isasymin}}\ \vals{A}}
\end{itemize}

In the following, we will need a way of piecing together infinite sequences.
Given two infinite sequences \isa{f} and \isa{g}, we can splice them at
position \isa{n}, s.t., in the resulting sequence all elements at positions
smaller than \isa{n} are taken from \isa{f} and all others are taken from
\isa{g}. This operation is written \isa{\repl{n}{f}{g}} and defined by
\isa{\repl{n}{f}{g}\ {\isaliteral{5C3C65717569763E}{\isasymequiv}}\ {\isaliteral{5C3C6C616D6264613E}{\isasymlambda}}j{\isaliteral{2E}{\isachardot}}\ \textsf{if}\ n\ {\isaliteral{5C3C6C653E}{\isasymle}}\ j\ \textsf{then}\ g\ j\ \textsf{else}\ f\ j}.

Furthermore, we say that an infinite sequence \isa{f} is \emph{minimal at} a
position \isa{n}, written \isa{\minat{f}{n}}, if all ``subsequences'' of
\isa{f} that coincide on the first \isa{n\ {\isaliteral{2D}{\isacharminus}}\ {\isadigit{1}}} elements and have a
smaller (w.r.t.\ \isa{\weaklt}) \isa{n}-th element are good (w.r.t.\ \isa{\strong}). The sense in which we use ``subsequence'' here, is made clear by
the following definition:
\begin{isadisplay}
\isa{\minat{f}{n}\ {\isaliteral{5C3C65717569763E}{\isasymequiv}}\ {\isaliteral{5C3C666F72616C6C3E}{\isasymforall}}g{\isaliteral{2E}{\isachardot}}\ {\isaliteral{28}{\isacharparenleft}}{\isaliteral{5C3C666F72616C6C3E}{\isasymforall}}i{\isaliteral{3C}{\isacharless}}n{\isaliteral{2E}{\isachardot}}\ g\ i\ {\isaliteral{3D}{\isacharequal}}\ f\ i{\isaliteral{29}{\isacharparenright}}\ {\isaliteral{5C3C616E643E}{\isasymand}}\ g\ n\ \weaklt\ f\ n\ {\isaliteral{5C3C616E643E}{\isasymand}}\ {\isaliteral{28}{\isacharparenleft}}{\isaliteral{5C3C666F72616C6C3E}{\isasymforall}}i{\isaliteral{5C3C67653E}{\isasymge}}n{\isaliteral{2E}{\isachardot}}\ {\isaliteral{5C3C6578697374733E}{\isasymexists}}j{\isaliteral{5C3C67653E}{\isasymge}}n{\isaliteral{2E}{\isachardot}}\ g\ i\ \weakle\ f\ j{\isaliteral{29}{\isacharparenright}}\ {\isaliteral{5C3C6C6F6E6772696768746172726F773E}{\isasymlongrightarrow}}\ \good{\strong}{g}}
\end{isadisplay}
which makes sure that objects in \isa{g} only contain elements that where
already present in some object of \isa{f}.

Now the key lemma in the construction of a minimal bad sequence is the
following:
\begin{lemma}\label{lem:step}\mbox{}
\begin{isadisplay}
\isa{{\isaliteral{5C3C6C6272616B6B3E}{\isasymlbrakk}}f\ {\isaliteral{28}{\isacharparenleft}}\suc{n}{\isaliteral{29}{\isacharparenright}}\ {\isaliteral{5C3C696E3E}{\isasymin}}\ \vals{A}{\isaliteral{3B}{\isacharsemicolon}}\ \minat{f}{n}{\isaliteral{3B}{\isacharsemicolon}}\ \bad{\strong}{f}{\isaliteral{5C3C726272616B6B3E}{\isasymrbrakk}}\isanewline
{\isaliteral{5C3C4C6F6E6772696768746172726F773E}{\isasymLongrightarrow}}\ {\isaliteral{5C3C6578697374733E}{\isasymexists}}g{\isaliteral{2E}{\isachardot}}\ {\isaliteral{28}{\isacharparenleft}}{\isaliteral{5C3C666F72616C6C3E}{\isasymforall}}i{\isaliteral{5C3C6C653E}{\isasymle}}n{\isaliteral{2E}{\isachardot}}\ g\ i\ {\isaliteral{3D}{\isacharequal}}\ f\ i{\isaliteral{29}{\isacharparenright}}\ {\isaliteral{5C3C616E643E}{\isasymand}}\isanewline
\isaindent{{\isaliteral{5C3C4C6F6E6772696768746172726F773E}{\isasymLongrightarrow}}\ {\isaliteral{5C3C6578697374733E}{\isasymexists}}g{\isaliteral{2E}{\isachardot}}\ }g\ {\isaliteral{28}{\isacharparenleft}}\suc{n}{\isaliteral{29}{\isacharparenright}}\ \weakle\ f\ {\isaliteral{28}{\isacharparenleft}}\suc{n}{\isaliteral{29}{\isacharparenright}}\ {\isaliteral{5C3C616E643E}{\isasymand}}\isanewline
\isaindent{{\isaliteral{5C3C4C6F6E6772696768746172726F773E}{\isasymLongrightarrow}}\ {\isaliteral{5C3C6578697374733E}{\isasymexists}}g{\isaliteral{2E}{\isachardot}}\ }{\isaliteral{28}{\isacharparenleft}}{\isaliteral{5C3C666F72616C6C3E}{\isasymforall}}i{\isaliteral{5C3C67653E}{\isasymge}}\suc{n}{\isaliteral{2E}{\isachardot}}\ {\isaliteral{5C3C6578697374733E}{\isasymexists}}j{\isaliteral{5C3C67653E}{\isasymge}}\suc{n}{\isaliteral{2E}{\isachardot}}\ g\ i\ \weakle\ f\ j{\isaliteral{29}{\isacharparenright}}\ {\isaliteral{5C3C616E643E}{\isasymand}}\isanewline
\isaindent{{\isaliteral{5C3C4C6F6E6772696768746172726F773E}{\isasymLongrightarrow}}\ {\isaliteral{5C3C6578697374733E}{\isasymexists}}g{\isaliteral{2E}{\isachardot}}\ }\bad{\strong}{\repl{\suc{n}}{f}{g}}\ {\isaliteral{5C3C616E643E}{\isasymand}}\ \minat{\repl{\suc{n}}{f}{g}}{\suc{n}}}
\end{isadisplay}
\end{lemma}
\begin{proof}
Let $P(f)$ abbreviate the conclusion of the key lemma.
We use the well-founded induction principle induced by the well-foundedness of
\isa{\weaklt} on the term \isa{f\ {\isaliteral{28}{\isacharparenleft}}\suc{n}{\isaliteral{29}{\isacharparenright}}}. As a result, we have to show
$P(g)$ for some arbitrary but fixed sequence $g$. We proceed by a case analysis
on whether $g$ is already minimal at position \isa{\suc{n}} or not. For
details, check \cite{afp-wqo}.
\end{proof}

By \lemref{step} we obtain a bad sequence that is minimal at \isa{\suc{n}}
from a bad sequence that is minimal at \isa{n}. This allows us to inductively
define a (globally) minimal bad sequence. The only missing part is that there
actually is a bad sequence that is minimal at \isa{{\isadigit{0}}}, which is shown
by the following lemma:
\begin{lemma}\label{lem:base}
\begin{isadisplay}
\isa{{\isaliteral{5C3C6C6272616B6B3E}{\isasymlbrakk}}f\ {\isadigit{0}}\ {\isaliteral{5C3C696E3E}{\isasymin}}\ \vals{A}{\isaliteral{3B}{\isacharsemicolon}}\ \bad{\strong}{f}{\isaliteral{5C3C726272616B6B3E}{\isasymrbrakk}}\ {\isaliteral{5C3C4C6F6E6772696768746172726F773E}{\isasymLongrightarrow}}\ {\isaliteral{5C3C6578697374733E}{\isasymexists}}g{\isaliteral{2E}{\isachardot}}\ {\isaliteral{28}{\isacharparenleft}}{\isaliteral{5C3C666F72616C6C3E}{\isasymforall}}i{\isaliteral{2E}{\isachardot}}\ {\isaliteral{5C3C6578697374733E}{\isasymexists}}j{\isaliteral{2E}{\isachardot}}\ g\ i\ \weakle\ f\ j{\isaliteral{29}{\isacharparenright}}\ {\isaliteral{5C3C616E643E}{\isasymand}}\ \minat{g}{{\isadigit{0}}}\ {\isaliteral{5C3C616E643E}{\isasymand}}\ \bad{\strong}{g}}
\end{isadisplay}
\end{lemma}
\begin{proof}
We use the same techniques as in the proof of \lemref{step}, but the second part
of the case analysis is considerably simpler.
\end{proof}

Finally we are in a position to show the existence of a minimal bad sequence
over objects constructed from elements of \isa{A}.
\begin{theorem}\label{lem:mbs}\mbox{}
\begin{isadisplay}
\isa{{\isaliteral{5C3C6C6272616B6B3E}{\isasymlbrakk}}{\isaliteral{5C3C666F72616C6C3E}{\isasymforall}}i{\isaliteral{2E}{\isachardot}}\ f\ i\ {\isaliteral{5C3C696E3E}{\isasymin}}\ \vals{A}{\isaliteral{3B}{\isacharsemicolon}}\ \bad{\strong}{f}{\isaliteral{5C3C726272616B6B3E}{\isasymrbrakk}}\ {\isaliteral{5C3C4C6F6E6772696768746172726F773E}{\isasymLongrightarrow}}\ {\isaliteral{5C3C6578697374733E}{\isasymexists}}g{\isaliteral{2E}{\isachardot}}\ \bad{\strong}{g}\ {\isaliteral{5C3C616E643E}{\isasymand}}\ {\isaliteral{28}{\isacharparenleft}}{\isaliteral{5C3C666F72616C6C3E}{\isasymforall}}n{\isaliteral{2E}{\isachardot}}\ \minat{g}{n}{\isaliteral{29}{\isacharparenright}}\ {\isaliteral{5C3C616E643E}{\isasymand}}\ {\isaliteral{28}{\isacharparenleft}}{\isaliteral{5C3C666F72616C6C3E}{\isasymforall}}i{\isaliteral{2E}{\isachardot}}\ g\ i\ {\isaliteral{5C3C696E3E}{\isasymin}}\ \vals{A}{\isaliteral{29}{\isacharparenright}}}
\end{isadisplay}
\end{theorem}
\begin{proof}
By \lemref{step} (which holds for every \isa{f} and \isa{n}) and the Axiom
of Choice, we obtain a choice function $\nu$, s.t., $\nu(f, n)$ yields the
corresponding witness. Moreover, by \lemref{base} we obtain a bad sequence
\isa{g} that is minimal at \isa{{\isadigit{0}}}. This allows us to define the
auxiliary sequence of sequences \isa{m{\isaliteral{27}{\isacharprime}}} recursively by:
\begin{align*%
}
m'(0) &= \nu(g, 0) \\
m'(n+1) &= \repl{n+1}{m'(n)}{\nu(m'(n), n)}
\end{align*%
}
The actual minimal bad sequence is then $m(i) = m'(i)(i)$. And the proof is
(mainly) shown by induction over $i$ (after considerably strengthening the
induction hypothesis). Again, we refer to \cite{afp-wqo} for the gory details.
\end{proof}%
\end{isamarkuptext}%
\isamarkuptrue%
\isadelimproof
\endisadelimproof
\isatagproof
\endisatagproof
{\isafoldproof}%
\isadelimproof
\endisadelimproof
\isamarkupsection{\label{sec:appl}Applications%
}
\isamarkuptrue%
\begin{isamarkuptext}%
Our current applications for the \isa{mbs} locale are \lemref{higman} and
\thmref{kruskal} (where the former is used in the proof of the latter). Thus, we
have to interpret the locale once in the context of lists and once in the
context of finite trees. For the latter we use the datatype
\begin{isadisplay}
\isa{\isacommand{datatype}\ \isasymalpha\ tree\ {\isaliteral{3D}{\isacharequal}}\ Node\ \isasymalpha\ {\isaliteral{28}{\isacharparenleft}}\isasymalpha\ tree\ list{\isaliteral{29}{\isacharparenright}}}
\end{isadisplay}
representing finite (non-empty) trees (isomorphic to ground terms of first-order
term rewriting as used in \cite{ceta}, where we want to apply our formalization
of wqo theory eventually).%
\end{isamarkuptext}%
\isamarkuptrue%
\begin{isamarkuptext}%
The three parameters for the list case are \isa{lists{\isaliteral{5C3C436F6C6F6E3E}{\isasymColon}}\isasymalpha\ set\ {\isaliteral{5C3C52696768746172726F773E}{\isasymRightarrow}}\ \isasymalpha\ list\ set}
(the set of lists over elements from a given set), \isa{\lemb{\isaliteral{5C3C436F6C6F6E3E}{\isasymColon}}\isasymalpha\ list\ {\isaliteral{5C3C52696768746172726F773E}{\isasymRightarrow}}\ \isasymalpha\ list\ {\isaliteral{5C3C52696768746172726F773E}{\isasymRightarrow}}\ bool} (homeomorphic
embedding on lists) and \isa{\suffix} (the suffix relation on lists), where we
use the following definitions:
\begin{definition}
The \emph{homeomorphic embedding} on lists, w.r.t.\ an arbitrary relation \isa{\wqole} on list elements, is defined inductively by
\begin{isadisplay}
\isa{{\isaliteral{5B}{\isacharbrackleft}}{\isaliteral{5D}{\isacharbrackright}}\ \lemb\ ys\isasep\isanewline%
xs\ \lemb\ ys\ {\isaliteral{5C3C4C6F6E6772696768746172726F773E}{\isasymLongrightarrow}}\ xs\ \lemb\ y\ {\isaliteral{23}{\isacharhash}}\ ys\isasep\isanewline%
{\isaliteral{5C3C6C6272616B6B3E}{\isasymlbrakk}}x\ \wqole\ y{\isaliteral{3B}{\isacharsemicolon}}\ xs\ \lemb\ ys{\isaliteral{5C3C726272616B6B3E}{\isasymrbrakk}}\ {\isaliteral{5C3C4C6F6E6772696768746172726F773E}{\isasymLongrightarrow}}\ x\ {\isaliteral{23}{\isacharhash}}\ xs\ \lemb\ y\ {\isaliteral{23}{\isacharhash}}\ ys}
\end{isadisplay}
which essentially says that we are allowed to drop elements or replace elements
by smaller ones (w.r.t., \isa{\wqole}) when going from right to left.
The \emph{suffix relation} on lists is given by
\begin{isadisplay}
\isa{xs\ \suffix\ ys\ {\isaliteral{5C3C65717569763E}{\isasymequiv}}\ {\isaliteral{5C3C6578697374733E}{\isasymexists}}us{\isaliteral{2E}{\isachardot}}\ ys\ {\isaliteral{3D}{\isacharequal}}\ us\ {\isaliteral{40}{\isacharat}}\ xs\ {\isaliteral{5C3C616E643E}{\isasymand}}\ us\ {\isaliteral{5C3C6E6F7465713E}{\isasymnoteq}}\ {\isaliteral{5B}{\isacharbrackleft}}{\isaliteral{5D}{\isacharbrackright}}}
\end{isadisplay}
\end{definition}
In order to obtain a \isa{\suffix}-minimal \isa{\lemb}-bad sequence we
have to prove the properties:
\begin{isadisplay}
\isa{{\isaliteral{5C3C6C6272616B6B3E}{\isasymlbrakk}}xs\ \lemb\ ys{\isaliteral{3B}{\isacharsemicolon}}\ ys\ \suffix\ zs{\isaliteral{5C3C726272616B6B3E}{\isasymrbrakk}}\ {\isaliteral{5C3C4C6F6E6772696768746172726F773E}{\isasymLongrightarrow}}\ xs\ \lemb\ zs} \\
\isa{\wfon{\suffix}{\lists{A}}} \\
\isa{{\isaliteral{5C3C6C6272616B6B3E}{\isasymlbrakk}}xs\ \suffix\ ys{\isaliteral{3B}{\isacharsemicolon}}\ ys\ \suffix\ zs{\isaliteral{5C3C726272616B6B3E}{\isasymrbrakk}}\ {\isaliteral{5C3C4C6F6E6772696768746172726F773E}{\isasymLongrightarrow}}\ xs\ \suffix\ zs} \\
\isa{{\isaliteral{5C3C6C6272616B6B3E}{\isasymlbrakk}}xs\ \suffix\ ys{\isaliteral{3B}{\isacharsemicolon}}\ ys\ {\isaliteral{5C3C696E3E}{\isasymin}}\ \lists{A}{\isaliteral{5C3C726272616B6B3E}{\isasymrbrakk}}\ {\isaliteral{5C3C4C6F6E6772696768746172726F773E}{\isasymLongrightarrow}}\ xs\ {\isaliteral{5C3C696E3E}{\isasymin}}\ \lists{A}}
\end{isadisplay}
all of which are easy.%
\end{isamarkuptext}%
\isamarkuptrue%
\begin{isamarkuptext}%
For the tree case, the parameters are \isa{\trees[Node]{A}{\isaliteral{5C3C436F6C6F6E3E}{\isasymColon}}\isasymalpha\ tree\ set} (the set of trees
over elements from a given set), \isa{\hemb{\isaliteral{5C3C436F6C6F6E3E}{\isasymColon}}\isasymalpha\ tree\ {\isaliteral{5C3C52696768746172726F773E}{\isasymRightarrow}}\ \isasymalpha\ tree\ {\isaliteral{5C3C52696768746172726F773E}{\isasymRightarrow}}\ bool} (homeomorphic embedding on trees) and \isa{\subtree[Node]} (the subtree relation; similar to the subterm relation on terms), where
we use the following definitions:
\begin{definition}
The \emph{set of trees} over elements from \isa{A} is given by
\begin{isadisplay}
\isa{{\isaliteral{5C3C6C6272616B6B3E}{\isasymlbrakk}}x\ {\isaliteral{5C3C696E3E}{\isasymin}}\ A{\isaliteral{3B}{\isacharsemicolon}}\ ts\ {\isaliteral{5C3C696E3E}{\isasymin}}\ \treeslist[Node]{A}{\isaliteral{5C3C726272616B6B3E}{\isasymrbrakk}}\ {\isaliteral{5C3C4C6F6E6772696768746172726F773E}{\isasymLongrightarrow}}\ Node\ x\ ts\ {\isaliteral{5C3C696E3E}{\isasymin}}\ \trees[Node]{A}\isasep\isanewline%
{\isaliteral{5B}{\isacharbrackleft}}{\isaliteral{5D}{\isacharbrackright}}\ {\isaliteral{5C3C696E3E}{\isasymin}}\ \treeslist[Node]{A}\isasep\isanewline%
{\isaliteral{5C3C6C6272616B6B3E}{\isasymlbrakk}}t\ {\isaliteral{5C3C696E3E}{\isasymin}}\ \trees[Node]{A}{\isaliteral{3B}{\isacharsemicolon}}\ ts\ {\isaliteral{5C3C696E3E}{\isasymin}}\ \treeslist[Node]{A}{\isaliteral{5C3C726272616B6B3E}{\isasymrbrakk}}\ {\isaliteral{5C3C4C6F6E6772696768746172726F773E}{\isasymLongrightarrow}}\ t\ {\isaliteral{23}{\isacharhash}}\ ts\ {\isaliteral{5C3C696E3E}{\isasymin}}\ \treeslist[Node]{A}}
\end{isadisplay}
\emph{Homomorphic embedding} on trees, w.r.t.\ an arbitrary relation \isa{\wqole}
on tree elements, is defined inductively by
\begin{isadisplay}
\isa{t\ {\isaliteral{5C3C696E3E}{\isasymin}}\ set\ ts\ {\isaliteral{5C3C4C6F6E6772696768746172726F773E}{\isasymLongrightarrow}}\ t\ \hemb\ Node\ f\ ts} \\
\isa{{\isaliteral{5C3C6C6272616B6B3E}{\isasymlbrakk}}f\ \wqole\ g{\isaliteral{3B}{\isacharsemicolon}}\ ss\ \embhemb\ ts{\isaliteral{5C3C726272616B6B3E}{\isasymrbrakk}}\ {\isaliteral{5C3C4C6F6E6772696768746172726F773E}{\isasymLongrightarrow}}\ Node\ f\ ss\ \hemb\ Node\ g\ ts} \\
\isa{{\isaliteral{5C3C6C6272616B6B3E}{\isasymlbrakk}}s\ \hemb\ t{\isaliteral{3B}{\isacharsemicolon}}\ t\ \hemb\ u{\isaliteral{5C3C726272616B6B3E}{\isasymrbrakk}}\ {\isaliteral{5C3C4C6F6E6772696768746172726F773E}{\isasymLongrightarrow}}\ s\ \hemb\ u} \\
\isa{s\ \hemb\ t\ {\isaliteral{5C3C4C6F6E6772696768746172726F773E}{\isasymLongrightarrow}}\ Node\ f\ {\isaliteral{28}{\isacharparenleft}}ss\ {\isaliteral{40}{\isacharat}}\ s\ {\isaliteral{23}{\isacharhash}}\ ts{\isaliteral{29}{\isacharparenright}}\ \hemb\ Node\ f\ {\isaliteral{28}{\isacharparenleft}}ss\ {\isaliteral{40}{\isacharat}}\ t\ {\isaliteral{23}{\isacharhash}}\ ts{\isaliteral{29}{\isacharparenright}}}
\end{isadisplay}
The \emph{subtree relation} is defined inductively by
\begin{isadisplay}
\isa{t\ {\isaliteral{5C3C696E3E}{\isasymin}}\ set\ ts\ {\isaliteral{5C3C4C6F6E6772696768746172726F773E}{\isasymLongrightarrow}}\ \subtreebin[Node]{t}{Node\ x\ ts}\isasep\isanewline%
{\isaliteral{5C3C6C6272616B6B3E}{\isasymlbrakk}}\subtreebin[Node]{s}{t}{\isaliteral{3B}{\isacharsemicolon}}\ t\ {\isaliteral{5C3C696E3E}{\isasymin}}\ set\ ts{\isaliteral{5C3C726272616B6B3E}{\isasymrbrakk}}\ {\isaliteral{5C3C4C6F6E6772696768746172726F773E}{\isasymLongrightarrow}}\ \subtreebin[Node]{s}{Node\ x\ ts}}
\end{isadisplay}
\end{definition}
In order to obtain a \isa{\subtree[Node]}-minimal \isa{\hemb}-bad sequence we
have to prove the properties:
\begin{isadisplay}
\isa{{\isaliteral{5C3C6C6272616B6B3E}{\isasymlbrakk}}s\ \hemb\ t{\isaliteral{3B}{\isacharsemicolon}}\ \subtreebin[Node]{t}{u}{\isaliteral{5C3C726272616B6B3E}{\isasymrbrakk}}\ {\isaliteral{5C3C4C6F6E6772696768746172726F773E}{\isasymLongrightarrow}}\ s\ \hemb\ u} \\
\isa{\wfon{\subtree[Node]}{\trees[Node]{A}}} \\
\isa{{\isaliteral{5C3C6C6272616B6B3E}{\isasymlbrakk}}\subtreebin[Node]{s}{t}{\isaliteral{3B}{\isacharsemicolon}}\ \subtreebin[Node]{t}{u}{\isaliteral{5C3C726272616B6B3E}{\isasymrbrakk}}\ {\isaliteral{5C3C4C6F6E6772696768746172726F773E}{\isasymLongrightarrow}}\ \subtreebin[Node]{s}{u}} \\
\isa{{\isaliteral{5C3C6C6272616B6B3E}{\isasymlbrakk}}\subtreebin[Node]{s}{t}{\isaliteral{3B}{\isacharsemicolon}}\ t\ {\isaliteral{5C3C696E3E}{\isasymin}}\ \trees[Node]{A}{\isaliteral{5C3C726272616B6B3E}{\isasymrbrakk}}\ {\isaliteral{5C3C4C6F6E6772696768746172726F773E}{\isasymLongrightarrow}}\ s\ {\isaliteral{5C3C696E3E}{\isasymin}}\ \trees[Node]{A}}
\end{isadisplay}
all of which are relatively easy.

Our bigger concern was (and rightly so) whether our definition of homeomorphic
embedding on trees really corresponds to what is usually used in the literature.
To this end, we used the definition of homeomorphic embedding that is used in
term rewriting (for first-order terms) as our specification.

\begin{definition}\label{def:Emb}
Let $\Emb$ be the (infinite) term rewrite system consisting of the following
rules:
\begin{align*%
}
f(\mathit{ts}) &\to t &&\text{if $t \in \mathit{set}~\mathit{ts}$}\\
f(\mathit{ts}) &\to g(\mathit{ss}) &&\text{if $g \wqole f$ and $\mathit{ss}
=_{\mathsf{emb}} \mathit{ts}$}
\end{align*%
}
\end{definition}
We were able (after adapting our initial inductive definition several times) to
prove%
\footnote{See theory \texttt{Embedding\_Trs} in the \isafor{} repository
\url{http://cl-informatik.uibk.ac.at/software/ceta}.}
that
\[
s \hemb t \longleftrightarrow t \to^+_{\Emb} s
\]
which reassures us that our definition is correct. (Note that for this proof we
had to use the datatype \isa{\isacommand{datatype}\ {\isaliteral{28}{\isacharparenleft}}\isasymalpha{\isaliteral{2C}{\isacharcomma}}\ \isasymbeta{\isaliteral{29}{\isacharparenright}}\ term\ {\isaliteral{3D}{\isacharequal}}\ Var\ \isasymbeta\ {\isaliteral{7C}{\isacharbar}}\ Fun\ \isasymalpha\ {\isaliteral{28}{\isacharparenleft}}{\isaliteral{28}{\isacharparenleft}}\isasymalpha{\isaliteral{2C}{\isacharcomma}}\ \isasymbeta{\isaliteral{29}{\isacharparenright}}\ term\ list{\isaliteral{29}{\isacharparenright}}} instead of \isa{\isasymalpha\ tree}.
However, it is easy to see that those two datatypes are isomorphic when
disregarding variables.)

\defref{Emb} is an adaption of the (potentially finite) term rewrite system from
\cite[Definition 3.4]{simple}, where we reuse the definition of homeomorphic
embedding on lists and avoid variables in order to simplify the above proof.%
\end{isamarkuptext}%
\isamarkuptrue%
\isamarkupsection{\label{sec:concl}Conclusions and Related Work%
}
\isamarkuptrue%
\begin{isamarkuptext}%
There have been formalizations of Higman's lemma in other proof assistants
\cite{murthy,fridlender,herbelin,seisenberger,martin-mateos} and also in
Isabelle/HOL \cite{berghofer} (but restricted to a two-letter alphabet of list
elements). Those existing formalizations usually strive for constructive proofs,
whereas our approach is purely classical. However, to the best of our knowledge
our work \cite{afp-wqo} constitutes the first formalization of Kruskal's tree
theorem ever.

Furthermore, \cite{myhill-nerode} could already make use of our formalization of
Higman's lemma for formalizing the following: \textsl{For an arbitrary language
$L$, the set of substrings/superstrings of words in $L$ is regular}.

A final remark, during the whole formalization process one of the key points was
to go away from proving facts (about binary predicates) on whole types and
instead make the carrier explicit. To this end, predicates like \isa{\reflon{\wqole}{A}}, \isa{\transon{\wqole}{A}}, \isa{\wfon{\wqole}{A}}, \ldots have been most
helpful and we plead to include them in the standard Isabelle/HOL distribution
and introduce their ``implicit'' cousins (working on whole types) as
abbreviations.

\paragraph{Acknowledgments.}
Thanks to Mizuhito Ogawa for helpful discussion and pointing out that the suffix
relation is appropriate for constructing a minimal bad sequence inside the proof
of Higman's lemma.
\bibliography{references}%
\end{isamarkuptext}%
\isamarkuptrue%
\isadelimtheory
\endisadelimtheory
\isatagtheory
\endisatagtheory
{\isafoldtheory}%
\isadelimtheory
\endisadelimtheory
\end{isabellebody}%

\end{document}